\documentclass[11pt]{amsart}
\usepackage{amsbsy,amssymb,amscd,amsfonts,latexsym,amstext,delarray,
amsmath,graphicx} 
\input xypic

\usepackage{color}

\newtheorem{thm}{Theorem}[section]
\newtheorem{prop}[thm]{Proposition}
\newtheorem{cor}[thm]{Corollary}
\newtheorem{lem}[thm]{Lemma}

\newtheorem{defn}[thm]{Definition}
\newtheorem{rem}[thm]{Remark}

\numberwithin{equation}{section}

\def\bF{{\mathbb F}}
\def\bG{{\mathbb G}}

\def\bL{{\mathbb L}}

\def\A{{\mathbb A}}
\def\C{{\mathbb C}}
\def\F{{\mathbb F}}

\def\N{{\mathbb N}}
\renewcommand{\P}{{\mathbb P}}
\def\Q{{\mathbb Q}}
\def\R{{\mathbb R}}
\def\Z{{\mathbb Z}}

\def\cA{{\mathcal A}}
\def\cB{{\mathcal B}}
\def\cC{{\mathcal C}}

\def\cH{{\mathcal H}}

\def\cJ{{\mathcal J}}

\def\cO{{\mathcal O}}

\def\cR{{\mathcal R}}

\def\cT{{\mathcal T}}

\def\cV{{\mathcal V}}
\def\cW{{\mathcal W}}

\def\End{{\rm End}}

\def\Tr{{\rm Tr}}
\def\tr{{\rm tr}}

\title[$q$-Deformations, QSM, and Motives]{$q$-Deformations of 
Statistical Mechanical Systems and Motives over finite fields}
\author{Matilde Marcolli and Zhi Ren}
\address{Department of Mathematics, California Institute of Technology, 1200 E California Blvd,
Pasadena, CA 91125, USA}
\email{matilde@caltech.edu}
\email{zren@caltech.edu}

\begin{document}
\maketitle

\begin{abstract}
We consider $q$-deformations of Witt rings, based on geometric operations on
zeta functions of motives over finite fields, and we use these deformations to
construct $q$-analogs of the Bost--Connes quantum statistical mechanical
system. We show that the $q$-deformations obtained in this way can be
related to Habiro ring constructions of analytic functions over $\F_1$ and
to categorifications of Bost--Connes systems. 
\end{abstract}

\section{Introduction}
In the present paper, we study $q$-deformations of the Bost--Connes system,
their relation to classical $q$-analog constructions, and the properties of the resulting quantum
statistical mechanical systems, in relation to motives over finite fields.  In this introductory section we review briefly the main
terminology and notation, with reference to the appropriate literature, about the Bost--Connes system,
and about some relevant notions of $q$-deformation.

\smallskip

In the rest of the paper, we first consider the $q$-deformations of the Witt ring $W(A)$ introduced in 
\cite{Oh1}, \cite{Oh2}, \cite{Oh3}, \cite{Oh4} and we show that, under these
deformations $W^q(A)$ of the Witt ring, the Bost--Connes algebra remains
undeformed. We then introduce some natural modifications of these $q$-deformations,
which are motivated by natural geometric operations on zeta functions of schemes of finite
type over $\F_q$. We show that these modified $q$-deformations give rise to $q$-deformed
Bost--Connes algebras. The first geometric deformation we consider corresponds, at the
level of zeta functions, to a geometric tower obtained by taking products of a given 
scheme $X$ with sets of $q^\ell$-points. This deformation is similar to the
one considered in \cite{Oh1}, \cite{Oh2}, \cite{Oh3}, \cite{Oh4}, but it deforms $W_0(A)$
leaving the product in the Witt ring $W(A)$ and $\Lambda(A)$ undeformed. The
resulting $q$-deformation of the Bost--Connes algebra, in this case, consists of
only a mild change that replaces the original integral Bost--Connes algebra $\cA_\Z$
of \cite{CCM3} with the product $\cA_{\Z}\otimes_\Z \Z[q]$. We then show that a more 
interesting $q$-deformation is obtained if, instead of taking products of schemes $X$
with zero-dimensional spaces consisting of $q^\ell$-points, one takes products with
affine spaces $\A^\ell$. Since $\# \A^\ell(\F_q)=q^\ell$, this deformation generalizes
the previous one, in the sense that the previous one occurs as first order term. 
We then show that, using these deformations one obtains interesting $q$-deformed
Bost--Connes algebras that can be related to the constructions of \cite{Man-hab},
\cite{Mar-endo}, and of \cite{LoMa}, and also to the categorifications of Bost--Connes
systems of \cite{MaTa}. We also discuss the role of $q$-analogs and a $q$-deformation
of the Riemann zeta function in this context. Finally, we propose the categorification
of Weil numbers and its associated quantum statistical mechanical system constructed
in \cite{MaTa} as another possible $q$-deformation of the Bost--Connes algebra that
fits into the general framework discussed in this paper.

\medskip
\subsection{The Bost--Connes quantum statistical mechanical system}

In \cite{BC}, Bost and Connes introduced a quantum statistical
mechanical system whose partition function is the Riemann zeta
function and whose equilibrium states are related to cyclotomic
fields with the Galois action. 
The algebra of observables of the system is a semigroup
crossed product algebra $C^*(\Q/\Z)\rtimes \N$, where
$C^*(\Q/\Z)=C(\hat\Z)$ is generated by abstract roots of unity $e(r)$, $r\in \Q/\Z$, and
the elements $n$ in the multiplicative semigroup are realized by isometries $\mu_n$. The 
semigroup $\N$ acts on $C^*(\Q/\Z)$ by endomorphisms
$$ \rho_n(e(r))=\mu_n e(r) \mu_n^* = \frac{1}{n} \sum_{s: ns=r} e(s). $$
The time evolution of the system is given by $\sigma_t(e(r))=e(r)$
and $\sigma_t(\mu_n)=n^{it}\, \mu_n$, with generator the Hamiltonian
$H\epsilon_k =\log(k)\epsilon_k$, in a representation on 
$\ell^2(\N)$. The partition function is the
Riemann zeta function,
$$ Z(\beta) = \Tr(e^{-\beta H}) = \sum_{k\geq 1} k^{-\beta}. $$
The extremal low temperature KMS states
are polylogarithm functions evaluated at roots of unity, normalized by
the Riemann zeta function, while at zero temperature the KMS states 
and the symmetry action are related to cyclotomic fields with the Galois action
of ${\rm Gal}(\Q^{ab}/\Q)$. We refer the reader to \cite{BC} and Chapter 3 of \cite{CoMa}
for a detailed discussion of these properties. 

The Bost--Connes system was reinterpreted in
\cite{CCM} as a particular example of an ``endomotive", a
projective system of zero-dimensional varieties with semigroup
actions. The analytical properties of the quantum statistical
mechanical system and its relation to $L$-functions were
also generalized to this context.
Moreover, in \cite{CCM3}, the Bost--Connes endomotive was
related to the geometry of the ``field with one element" $\F_1$ and
its system of extensions (given by roots of unity) as defined by
Kapranov--Smirnov \cite{KaSmi}.
In \cite{CoCo}, the integral Bost--Connes system considered
in \cite{CCM3} was related to Witt rings and the Frobenius
and Verschiebung operators. In \cite{MaTa}, generalizations and
categorifications of the Bost--Connes system were constructed, 
which include Weil numbers and motives over finite fields. 

Further aspects of the relation between the Bost--Connes system
and $\F_1$-geometry were considered in \cite{Man-hab} and 
\cite{Mar-endo}. Manin proposed in \cite{Man-hab} to use
the Habiro ring of $q$-functions, \cite{Hab}, as a good notion of 
analytic geometry and analytic functions over $\F_1$, and analogs
of the Bost--Connes system based on the Habiro ring were
constructed in \cite{Mar-endo}. Relations between the Habiro
ring and motives were discussed in \cite{LoMa}.

\medskip
\subsection{The integral Bost-Connes algebra}\label{intBCsec}

An integral model of the Bost-Connes algebra was introduced in \cite{CCM3},
in relation to $\F_1$-geometry. 

\smallskip

The integral Bost--Connes algebra $\mathcal{A}_\mathbb{Z}$ is the algebra generated by the group
ring $\Z[\Q/\Z]$ and by elements $\tilde\mu_n$ and $\mu_n^*$ satisfying the relations 
$$\tilde\mu_n x \mu_n^* = \tilde{\rho_n}(x)$$
$$\mu_n^* x = \sigma_n(x)\mu_n^*$$
$$ x \tilde\mu_n  = \tilde\mu_n \sigma_n(x),$$
for all $x\in \Z[\Q/\Z]$ and all $n\in \N$, 
where $\sigma_n(x)$ is an endomorphism of $\Z[\Q/\Z]$ given by $\sigma_n(e(r)) = e(nr)$, while  
$\tilde\rho_n (e(r)) = \sum_{nr' = r}e(r')$. The elements $\tilde\mu_n$ and $\mu_n^*$ also satisfy the relations
$$\tilde\mu_{nm} = \tilde \mu_n \tilde \mu_m, \ \ m, n\in \mathbb{N}, \ \ \ \
\mu_{nm}^* = \mu_n^* \mu_m^*, \ \ m, n\in \mathbb{N},$$
$$ \mu_n^*\tilde \mu_n = n, \ \   n \in \N, \ \ \ 
 \tilde \mu_n \mu_n^* = \mu_n^*\tilde \mu_n , \ \  (n,m) = 1.  $$

\smallskip

In \cite{CoCo}, the integral BC algebra is related to the 
universal Witt ring $W_0(\bar{\F}_p)$, where  $\bar{\F}_p$ is the algebraic 
closure of $\F_p$. Moreover, it is shown that 
that there is a $p$-adic representation of the integral BC algebra into the 
Big Witt Ring $W(\bar{\F}_p)$, determined by the embedding of $W_0(\bar{\F}_p)$
into  $W(\bar{\F}_p)$.

\medskip
\subsection{$q$-Deformations}

There is a vast literature on the subject of $q$-deformations and
$q$-analogs. The basic form of $q$-analog consists of the $q$-integers
$$ [n]_q := \frac{1-q^n}{1-q} = 1+q+\cdots + q^{n-1}. $$
Generalizations of the Riemann zeta function based on $q$-integers
have been considered, for instance, in \cite{Cher}, \cite{Kane}, \cite{Kawa}, 
and in a different form, suitable for an Euler product expansion, in \cite{Raw1},
\cite{Raw2}, while $q$-Bernoulli numbers were considered in \cite{Carl}.
A $q$-deformed analog of the polylogarithm function, the
$q$-polylogarithm, was considered in \cite{Zhao}. The Witt ring also
has a natural $q$-deformation, which was studied in \cite{Oh1},
\cite{Oh2}, \cite{Oh3}. These and other $q$-deformations of Witt rings will
be crucial to our analysis of possible $q$-deformations of the Bost--Connes
system. We discuss these $q$-deformations in the next section, and other
forms of ``geometric" $q$-deformations in the following sections.

\section{Witt Rings and their $q$-Deformations}\label{SecWittq}

\smallskip
\subsection{Witt rings and operations}\label{WittSec}

We first recall some basic definitions regarding Witt rings. For $A$ an associative commutative
ring, let $\underline{\End_A}$ be the category of endomorphisms of projective $A$-modules 
of finite rank. The objects of this category are pairs $(E,f)$, where
$f \in \End_A(E)$. With the direct sum and the tensor product defined componentwise on the objects, 
the Grothendieck group $K_0(\underline{\End_A})$ also acquires a commutative ring structure. 
Let $K_0(A)$ be the ideal generated by the pairs of the form $(E, f = 0)$. Then one defines 
\begin{equation}\label{W0A}
W_0(A) = K_0(\underline{\End_A})/K_0(A). 
\end{equation}
There are several operators defined on $W_0(A)$, including the Frobenius, 
the Verschiebung, and the Teichm\"uller and ghost maps, \cite{GA}. These play an
important role in the relation between Witt rings and the integral Bost--Connes
system, as shown in \cite{CoCo}. 

\smallskip

The ghost ring functor $gh: {\rm Rings} \rightarrow {\rm Rings}$ 
associates to each object $A$ the ring whose underlying
set is $$A^\infty:=\{ (x_n)_{n\in \N}\,|\, x_n \in A \}$$ and to each morphism 
$f: A \rightarrow B$ the morphism 
$$ {\rm gh}(f):(x_n)_{n\in \mathbb{N}} \mapsto (f(x_n))_{n\in \mathbb{N}}. $$

\smallskip

Given an associative and commutative ring A, 
the Big Witt Ring $W(A)$ is characterized by the three properties:
\begin{enumerate}
\item As a set, $W(A)$ equals $A^\infty$.
\item For any ring homomorphism $f: A \rightarrow B$, the map
$W(f): W(A) \rightarrow W(B)$ given by $$(x_n)_{n\in \mathbb{N}} \mapsto 
(f(x_n))_{n\in \mathbb{N}}$$
is a ring homomorphism. 
\item The map $\Phi: W(A)\rightarrow {\rm gh}(A)$, given by 
$$ (x_n)_{n\in \mathbb{N}} \mapsto (\sum_{d \vert n}d x_d^{\frac{n}{d}})_{n\in \mathbb{N}},$$
is a ring homomorphism, where ${\rm gh}(A)$ is the image of $A$ under the ghost ring
functor described above.
\end{enumerate}

The Frobenius, Verschiebung, Teichm\"uller, and ghost map operators 
are also defined on the ring $W(A)$, satisfying the same 
set of relations as on $W_0(A)$, see \cite{GA}. 

\smallskip

We also recall briefly the relation between $W_0(A)$ and $W(A)$, and the Grothendieck 
$\lambda$-ring $\Lambda(A)$, see \cite{GA}, \cite{Rama}.

\smallskip

Given an element $(E,f)$ in $W_0(A)$, 
let $M(f)$ denote the matrix representation of $f: E\to E$, as in \cite{GA}, \cite{CoCo}. 
The following properties determine a natural embedding of $W_0(A)$ into $W(A)$, see \cite{GA}.
\begin{itemize}
\item The map $L: W_0(A) \rightarrow \Lambda(A)$ given by $L(E,f) \mapsto \det(1-tM(f))^{-1}$ is a homomorphic injection into the additive group of the $\Lambda$ ring, with image given by the subgroup
\begin{equation}\label{rangeL}
 {\rm Range}(L) = \left\{ \frac{1+a_1 t+...+a_n t^n}{1+b_1 t+...+b_m t^m}, a_i, b_j \in A \right\} .
\end{equation} 
\item The Artin-Hasse exponential map $\varepsilon: W(A) \rightarrow \Lambda(A)$ 
given by $$\varepsilon: (x_n) \mapsto \prod_{n \in \mathbb{N}} \frac{1}{1-x_n t^n}$$ 
is an isomorphism.
\end{itemize}

\smallskip 
\subsection{q-Deformations of Witt rings}\label{qWittSec}

$q$-Deformations of the Witt Ring, considered as a $q$-deformed functor from the category 
of rings to itself, and of the ghost map were introduced in \cite{Oh1}, \cite{Oh2}, \cite{Oh3}.
We recall the basic properties of these deformations.

\smallskip

For $q$ a positive integer and $A$ an associative commutative ring, the $q$-deformed Big Witt Ring 
$W^q(A)$ is characterized in \cite{Oh3} by the following three conditions:
\begin{enumerate}
\item As a set, $W^q(A)$ equals $A^\infty$
\item For any ring homomorphism $f: A \rightarrow B$, the map $W^q(f): W^q(A) \rightarrow W^q(B)$
$$ (x_n)_{n\in \mathbb{N}} \mapsto (f(x_n))_{n\in \mathbb{N}}$$
is a ring homomorphism. 
\item The map $\Phi^q: W(A)\rightarrow gh(A)$
\begin{equation}\label{Phiqmap}
 (x_n)_{n\in \mathbb{N}} \mapsto (\sum_{d \vert n}d q^{\frac{n}{d} - 1} x_d^{\frac{n}{d}})_{n\in \mathbb{N}}
\end{equation} 
is a ring homomorphism. 
\end{enumerate}

\smallskip

Moreover, the Grothendieck $\lambda$-ring also admits a $q$-deformation $\Lambda^q(A)$, as in \cite{Oh2}.
As a set, $\Lambda^q(A) = 1 + tA[t] = \lbrace1+\sum_{n=1}^{\infty} a_n t^n, a_n \in A, n \geq 1 \rbrace$, 
where the addition is defined as the usual multiplication of power series raised to the $q$-th power, and 
the multiplication is defined by requiring that
\begin{equation}\label{qWittProd}
(1-at)^{-q} \star_q (1-bt)^{-q} = (1-abt)^{-q} .
\end{equation}

\smallskip

It is proved in \cite{Oh1} that, for every commutative ring in which $q$ is invertible, 
the map $\eta: \Lambda(A) \rightarrow \Lambda^q(A)$, given by $\eta(f) = f(t)^q$, is 
an isomorphism. Thus, in this case, one can identify the underlying set of $\Lambda^q(A)$ with the 
$q$-th powers of the power series in $t$ with coefficients
in $A$ having constant term equal to $1$.

\smallskip

Note that the product \eqref{qWittProd} under $\star_q$ in $\Lambda^q(A)$ differs from the 
usual product $(1-at)^{-q} \star (1-bt)^{-q} =(1-ab t)^{-q^2}$ as elements of $\Lambda(A)$. 
In particular, for $A=k$, while the product of $(1-at)^{-q}$ and $ (1-bt)^{-q}$ in $\Lambda(k)$
can be interpreted as $L((E_1^{\oplus q},f_1^{\oplus q})\otimes (E_2^{\oplus q},f_2^{\oplus q}))
=(1-ab t)^{-q^2}$ for $(E_1,f_1)=(k,a)$ and $(E_2,f_2)=(k,b)$, the corresponding identity
$$ L^q(k,a)\star_q L^q(k,b)=L^q(k,ab) $$
that matches \eqref{qWittProd} does not correspond to just identifying $L^q(E,f)$ with
$L(E^{\oplus q},f^{\oplus q})$.

\smallskip

\begin{prop}\label{qcharpol}
Let $A$ be a  commutative ring in which $q$ is invertible.
The $q$-deformed characteristic polynomial 
\begin{equation}\label{Lq}
 L^q: W_0(A) \rightarrow \Lambda^q(A),  \ \ \ \  L^q(E,f) = \det(1-tM(f))^{-q}
\end{equation} 
determines a homomorphic injection whose image is 
$$ {\rm Range}(L^q) = \lbrace \frac{(1+a_1 t+...+a_n t^n)^q}{(1+b_1 t+...+b_m t^m)^q}, a_i, b_j \in A \rbrace.$$
\end{prop}

\begin{proof}
The result of \cite{Oh1} mentioned above, showing that the map $\eta(f) = f(t)^q$ is 
an isomorphism, combined with the homomorphic injection $L(E,f) \mapsto \det(1-tM(f))^{-1}$ to \eqref{rangeL},
implies that the following diagram commutes 
$$ \begin{CD}
\Lambda(A) @>\eta>>  \Lambda^q(A)\\
@AALA                             @AAL^qA\\
W_0(A)   @<id<<       W_0(A)\, .
\end{CD} $$
The result then follows.
\end{proof}

\smallskip

\begin{cor}\label{embW0Wq}
The $q$-deformed Witt ring $W^q(A)$ contains an isomorphic copy of $W_0(A)$.
\end{cor}

\begin{proof} The $q$-deformed Witt ring $W^q(A)$ is isomorphic to 
$\Lambda^q(A)$ through a $q$-analog of the Artin--Hasse map, see \cite{Oh4}. It is then easy to see 
that $W_0(A)$ can be embedded inside $W^q(A)$, with 
the map given by the composition of the $q$-deformed characteristic polynomial \eqref{Lq}
with the $q$-analog of the Artin--Hasse map.  
\end{proof}

Later, we will describe an explicit canonical embedding of $W_0(A)$ into $W^q(A)$, in the special case where 
$A=k$ is an algebraically closed field. 

\smallskip
\section{Undeformed Bost--Connes algebras in $q$-deformed Witt rings}\label{qBCZSec}

The integral form of the Bost--Connes algebra introduced in \cite{CCM3} 
is directly related to the Witt ring $W_0(A)$, as shown in \cite{CoCo}. The
operations $\sigma_n$ and $\tilde\rho_n$ of the integral Bost--Connes
algebra are induced by the Frobenius and Verschiebung and their
extensions to the Witt ring $W(A)$. In this section, we consider the
$q$-embedding of $W_0(A)$ into $W^q(A)$ of Corollary~\ref{embW0Wq},
and the natural operations in $W^q(A)$, and we show that the same
construction of \cite{CoCo} goes over with minor modifications to this
case and realize the same (undeformed) integral Bost--Connes algebra in terms
of Frobenius and Verschiebung on the $q$-deformed Witt ring $W^q(A)$.
We will see in the following section how this suggests then a natural
deformation of the integral Bost--Connes algebra, based on a modification
of the $q$-deformation $W^q(A)$ of the Witt ring, which differs from the
one introduced in \cite{Oh1}, \cite{Oh2}, \cite{Oh3}, \cite{Oh4}. 

\smallskip
\subsection{Operations on q-deformed Witt rings}\label{opsqWitt}

We begin by checking that operations and relations on $W_0(A)$ 
extend compatibly to $W^q(A)$ through the $q$-embedding.

\smallskip

As mentioned in \S \ref{WittSec}, the ghost map, the Frobenius, and the Verschiebung 
are operators defined on $W_0(A)$,  and $W(A)$. It is shown in 
\cite{GA} that the operators defined on them are, in the appropriate sense, compatible. We now prove that 
the operators defined on $W_0(A)$, when the latter is $q$-embedded in $W^q(A)$, are also compatible 
with the operators defined on $W^q(A)$.

\begin{defn}
Let $A^q[t]$ be the ring whose underlying set is $A[t]$, with the addition defined by the usual addition of power series, and the
multiplication defined by 
\begin{equation}\label{Aqtmult}
 \sum a_n t^n \sum b_n t^n = \sum \frac{1}{q} (a_n b_n)t^n .
\end{equation}
\end{defn}

\smallskip

In the following we focus on the case where $A=k$ is an algebraically closed field. We obtain the
following compatibilities between $W_0(k)$ and $W^q(k)$ under the $q$-embedding. 

\begin{prop}
Let $A=k$ be an algebraically closed field. 
The following compatibilities hold between operations on $W_0(k)$ and $W^q(k)$ under the $q$-embedding
of Corollary~\ref{embW0Wq}.
\begin{enumerate}
\item the ghost map $gh_n$ defined on $W_0(k)$ by 
$$ gh_n: W_0(k) \rightarrow gh'(k), \ \ \ \  (E,f) \mapsto (\tr(f^n))_{n\in \mathbb{N}} $$ 
is compatible with the ghost map defined 
on $W^q(k)$ by $\Phi^q$ of \eqref{Phiqmap}, in the sense that the following diagram commutes:
$$\begin{CD}
W_0(k) @>L^q>>             \Lambda^q(k)                  @=                W^q(k)\\
@VVgh_nV                             @VV \frac{d}{dt} logV                    @VV\Phi^qV\\
gh'(k)   @>\iota^q>>       k^q[t]                         @=               gh(k) \, 
\end{CD}$$
where the $q$-identification $\iota^q$ is given by 
\begin{equation}\label{iotaq}
\iota^q: (x_n) \mapsto \sum_n q x_n t^{n-1}.
\end{equation}

\item
The Frobenius map $F_n$ defined on $W_0(k)$ by $(E,f) \rightarrow (E,f^n)$ is compatible with the Frobenius map $F_n'$ on $W^q(k)$ in the sense
that the following diagram is commutative:
$$\begin{CD}
W_0(k) @>L^q>>             \Lambda^q(k)                  @=                W^q(k)\\
@VVF_nV                             @VVF_n''V                                                 @VVF_n'V\\
W_0(k)   @>L^q>>        \Lambda^q(k)                       @=                W^q(k) \, .
\end{CD}$$

\item
The Verschiebung operator $V_n$ defined on $W_0(k)$ by 
\begin{equation}\label{Versch}
(E,f) \rightarrow \left( E^{\bigoplus n}, \left( \begin{matrix} 0 & 0 & ... & ... & f\\ 1 & 0 & 0 & ... & 0\\ \hdotsfor{5} \\ 0 & 0 & 0 & ... 1& 0 \end{matrix} \right) \right)  
\end{equation}
is compatible with the Verschiebung operator $V_n'$ on $W^q(k)$ in the sense
that the following diagram is commutative:
$$\begin{CD}
W_0(k) @>L^q>>             \Lambda^q(k)                  @=                W^q(k)\\
@VV{V_n}V                             @.                                                 @VV{V_n'} V\\
W_0(k)   @>L^q>>        \Lambda^q(k)                       @=                W^q(k) \, .
\end{CD}$$
\end{enumerate}
\end{prop}

\begin{proof}
We work under the assumption that $A = k$ is an algebraically closed field. 
\begin{enumerate}
\item
It is clear from \cite{GA} and \cite{Oh4} that the right half of the diagram is commutative. 
The isomorphism between $\Lambda^q(k)$ and $W^q(k)$ is given by the $q$-deformed 
Artin--Hasse exponential map. The isomorphism between $gh(k)$ and $k^q[t]$ is given 
by the $q$-identification \eqref{iotaq}.
In an algebraically closed field, the matrix associated with the endomorphism $f$ can be triangulated. 
Therefore, the $n$-th component of the ghost map is given by $\sum_i \alpha_i^n$, where $\alpha_i$ 
are the eigenvalues associated with the endomorphism $f$. Note that, given these eigenvalues, 
$$L^q(E,f) = \det(1-tf)^{-q} = \prod_i (1-\alpha_i)^{-q}.$$ Taking the log-derivative of this, we then obtain
$$\frac{d}{dt} \log(\prod_i (1-\alpha_i)^{-q}) = q \frac{d}{dt} \sum_i  \log(1 - \alpha_i)^{-1} 
= q \sum_j^\infty (\sum_i \alpha_i ^j) t^{j-1},$$ 
where the second identify follows from the identity 
$$\frac{d}{dt} \log(\frac{1}{1-at}) = \sum_i a^i t^{i-1},$$ which is obtained  
from the Taylor expansion of the log function. 
Given the $i$-th ghost component of $W_0(k)$, which is $x_i = \sum_i \alpha_i ^j$, we obtain through the 
$q$-identification, $$\iota^q((x_n)) = q \sum_j^\infty (\sum_i \alpha_i ^j) t^{j-1}. $$

\item
As in \cite{GA}, using the identification between $W^q(k)$ and the ring $\Lambda^q(k)$, we see that 
the Frobenius map on $W^q(k)$ is the same as the map
$$F_n'': (1-at)^{-1} \mapsto (1-a^n t)^{-1}$$ 
Then, suppose given $(E,f) \in W_0(k)$, where $f$ has eigenvalues $\alpha_i$. 
The Frobenius map sends $(E,f)$ to $(E,f^n)$, where the eigenvalues of $f^n$ 
ware $\alpha_i^n$. Then we obtain 
 $$F_n''(L^q(E,f)) = F_n''(\prod_i (1-\alpha_i)^{-q}) = 
 \prod_i (1-\alpha_i^n)^{-q} = L^q(E, f^n) = L^q(F_n(E,f)).$$
 
 \item Again, the equivalence of the right half of the diagram is shown in \cite{Oh4}.
Note that the Verschiebung operator $V_n$ acting on $\Lambda^q(k)$ is defined by 
$$ V_n:  (1-at)^{-1} \mapsto (1-at^n)^{-1}, $$
for $(1-at)^{-1} \in\Lambda^q(k)$. A direct calculation then shows that 
$$ \det \left( 1-  \left( \begin{matrix} 0 & 0 & ... & ... & f\\ 1 & 0 & 0 & ... & 0\\ \hdotsfor{5} \\ 0 & 0 & 0 & ... 1& 0 \end{matrix} \right) \, \, t \right)^{-1} = (1- M(f)\, t^n )^{-1}. $$
Thus, the commutativity of the left half of the diagram also follows. 
\end{enumerate}
\end{proof}

\smallskip
\subsection{Divisor map and undeformed Bost--Connes algebra}\label{DivSec}

When $A$ is an algebraically closed field $k$, 
the determinant function factors completely into linear forms in 
terms of the eigenvalues of the endomorphisms. 
Associating to each pair $(E,f)$ the divisor $\delta(f)$ of non-zero 
eigenvalues of $f$ determines a ring isomorphism between 
$W_0(k)$ and the group ring $\Z[k^*]$. If in particular $k=\bar\F_p$ is the 
algebraic closure of a finite field $\F_p$, it is shown in \cite{CoCo} that one 
obtains a natural isomorphism 
\begin{equation}\label{sigmabarFp}
 \sigma: W_0(\bar{\F}_p) \rightarrow \Z[(\Q/\Z)^{(p)}], 
\end{equation} 
where $(\Q/\Z)^{(p)}$ is the group of fractions with denominator prime to $p$. 

\smallskip

The map $\sigma$ of \eqref{sigmabarFp} is induced by the divisor map. 
Given an element $(E,f)$ in $W_0(k)$ with $\alpha$ the eigenvalues of $f$ and
$n(\alpha)$ their multiplicities, the divisor map is given by
\begin{equation}\label{divmap}
\delta(f):=\delta(L(E,f))=\delta (\prod (1-\alpha t)^{-n(\alpha)}) = \sum n(\alpha)  [\alpha], 
\end{equation}
As shown in Proposition 2.3 of \cite{CoCo}, it defines an element in $\Z[k^*]$,
and one obtains a ring isomorphism
\begin{equation}\label{divmap2}
\delta: W_0(k) \to \Z[k^*].
\end{equation}

\smallskip

When we consider $W_0(k)$ as embedded in the $q$-deformed Witt ring $W^q(k)$,
or equivalently we use the $q$-deformed characteristic polynomial $L^q: W_0(k)\to \Lambda^q(k)$,
the divisor map \eqref{divmap} is no longer compatible with the multiplication $\star_q$ in $\Lambda^q(k)$.
Indeed, we have
$$ \delta(L^q(E,f))=\delta(\prod (1-\alpha t)^{-q n(\alpha)}) =\sum q\, n(\alpha)\, [\alpha]. $$
Using $(k,a)\otimes (k,b) = (k, ab)$, 
we have $L^q(k,a)=(1-at)^{-q}$ and $L^q(k,b)=(1-bt)^{-q}$ in $\Lambda^q(k)$ with
$L^q((k,a)\otimes (k,b))=L^q(k,a)\star_q L^q(k,b)=(1-ab t)^{-q}=L^q(k, ab)$, by the definition
of the induced product $\star_q$ on $\Lambda^q(k)$, while the multiplication of the
divisors as elements in the group ring gives
$\delta(L^q(k,a))\delta(L^q(k,b)) = (\sum q\, n(\alpha)\, [\alpha])(\sum q \, n(\beta) [\beta])
= q \cdot \delta(L^q(k, ab))$. 

\smallskip

Thus, the only way to restore the multiplicative property is to compute 
$\delta_q(L^q(E,f))=q^{-1} \delta(L^q(E,f))=\delta(L(E,f))$. This restores the original
map \eqref{divmap} on the undeformed $\Lambda(k)$, hence the same undeformed
Bost--Connes algebra as constructed in \cite{CoCo}.

\smallskip

While the $q$-deformations of Witt rings considered in \cite{Oh1}, \cite{Oh2}, \cite{Oh3}, \cite{Oh4} do
not directly lead to a $q$-deformed Bost--Connes algebra through the same construction of \cite{CoCo},
this suggests a modification of the construction of $q$-deformed Witt rings, with a different motivation
in mind than the $q$-M\"obius functions and $q$-deformed necklace rings that motivated the 
construction of \cite{Oh1}, \cite{Oh2}, \cite{Oh3}, \cite{Oh4}. 

\medskip
\section{Geometric $q$-deformations of Witt rings} 

In this section we discuss a different approach to $q$-deforming the Witt rings, and we
show that, unlike the case discussed in the previous section, this leads to $q$-deformations
of the Bost--Connes algebra. The crucial difference here is that, instead of $q$-deforming
the Witt ring $W(A)$ to $W^q(A)$, or equivalently deforming $\Lambda(A)$ to $\Lambda^q(A)$
as in \cite{Oh1}, \cite{Oh2}, \cite{Oh3}, \cite{Oh4}, we consider a deformation of $W_0(A)$
to a $q$-deformed $\cW^q_0(A)$, while we maintain the product in $W(A)$ and $\Lambda(A)$
undeformed.
 
\smallskip

Geometrically, if we consider elements of $\Lambda(k)$ that arise from zeta functions of
schemes (see \cite{Rama}), the two kinds of deformations that we introduce
in this section have a very simple geometric meaning. The first arises by replacing a scheme 
$X$ with a tower where $X_{q^\ell}=X\sqcup \cdots \sqcup X$, a disjoint union of $q^\ell$
copies of $X$ (or equivalently the product of $X$ with a $q^\ell$ points). The second 
deformation consists of replacing $X$ with the tower of the products $X\times \A^\ell$.
Since, when $q$ is a prime power $q=p^r$, the number of points $\# A^\ell(\F_q)=q^\ell$,
the second construction will be an extension of the first, where the case of $q^\ell$-points
appears as the first term, but all the additional contributions of the field extensions
$\A^\ell(\F_{q^n})$ are also counted. Thus, we refer to them, respectively, as the
$q^{\ell}$-points deformation and the $\A^\ell$-deformation.

\smallskip
\subsection{The $q^{\ell}$-points deformation}

Let $\Omega_q$ denote the map $\Omega_q: W_0(A)\to W_0(A)$ that 
maps $\Omega_q: (E,f) \mapsto (E^{\oplus q}, f^{\oplus q})$. Consider
the restriction of the characteristic polynomial map
$L: W_0(A) \to \Lambda(A)$ to the range $\Omega_q(W_0(A))$.
Notice that, unlike the $q$-deformations considered in the previous
sections, here we do not deform the product in $\Lambda(A)$. We work with
the undeformed product determined by
$$ (1-at)^{-1} \star (1-bt)^{-1} = (1-ab t)^{-1}. $$
We define $S^q(E,f)=L(\Omega_q(E,f))$. Note that this is the same
characteristic polynomial $(1-tM(f))^{-q}$, as in the case of the
$q$-deformed characteristic polynomial $L^q(E,f)$ considered above,
except that now we regard it as an element of the undeformed ring $\Lambda(A)$
rather than as an element of $\Lambda^q(A)$. 

\smallskip

\begin{defn}\label{qellPtsDef}
The $q^\ell$-points deformation of the Witt ring $W_0(A)$ is 
the graded ring $\cW^q_0(A)$ defined as a set by
\begin{equation}\label{cWq}
\cW^q_0(A) =\oplus_{\ell \geq 0}  \Omega_{q^\ell}(W_0(A)), 
\end{equation}
with addition and multiplication induced uniquely by addition 
and multiplication in $W_0(A)$.
\end{defn}

\smallskip

\begin{lem}\label{cWmult}
The multiplication operation on $\cW^q_0(A)$ obtained as above satisfies
\begin{equation}\label{cWgradedmult}
 \star: \Omega_{q^\ell}(W_0(A)) \times \Omega_{q^{\ell'}} (W_0(A)) \to \Omega_{q^{\ell+\ell'}}(W_0(A))). 
\end{equation} 
\end{lem}

\proof We have $\Omega_{q^\ell}(E_1,f_1)=(E_1^{\otimes q^\ell}, f_1^{\otimes q^\ell})$ and 
$\Omega_{q^{\ell'}}(E_2,f_2)=(E_2^{\otimes q^{\ell'}}, f_2^{\otimes q^{\ell'}})$, hence their product
in $W_0(A)$ is given by $((E_1\otimes E_2)^{\oplus q^{\ell+\ell'}},
(f_1\otimes f_2)^{\oplus q^{\ell+\ell'}})$ and we obtain
\begin{equation}\label{prodOmegaq}
 \Omega_{q^\ell}(E_1,f_1) \otimes \Omega_{q^{\ell'}}(E_2,f_2) 
 = \Omega_{q^{\ell+\ell'}}(E_1\otimes E_2,f_1\otimes f_2). 
\end{equation} 
\endproof

At the level of characteristic polynomials this corresponds to the product relation
\begin{equation}\label{Sellprod}
S^{q^\ell}(E_1,f_1)\star S^{q^{\ell'}}(E_2,f_2) = S^{q^{\ell+\ell'}}(E_1\otimes E_2,f_1\otimes f_2)),
\end{equation}
since we have
$$  L(\Omega_{q^\ell}(E_1,f_1))\star L(\Omega_{q^{\ell'}}(E_2,f_2))) =
L(\Omega_{q^\ell}(E_1,f_1) \otimes \Omega_{q^{\ell'}}(E_2,f_2)) $$ 
$$=L(\Omega_{q^{\ell+\ell'}}(E_1\otimes E_2,f_1\otimes f_2)). $$

\smallskip

\begin{lem}\label{FnVncWq}
The Frobenius $F_n$ and Verschiebung $V_n$ on $W_0(A)$ extend to 
a Frobenius $F_n$ and Verschiebung $V_n$ on $\cW^q_0(A)$,
satisfying $\Omega_{q^\ell}\circ F_n = F_n \circ \Omega_{q^\ell}$ and
$\Omega_{q^\ell}\circ V_n = \tilde V_n \circ \Omega_{q^\ell}$, where
$\tilde V_n$ is $V_n$ up to a a change of basis given by a permutation. 
\end{lem}

\proof We have $F_n(E,f)=(E,f^n)$, hence $F_n(E^{\oplus q^\ell}, f^{\oplus q^\ell})=
(E^{\oplus q^\ell}, (f^n)^{\oplus q^\ell})=\Omega_{q^{\ell}}(E,f^n)$. We index
the entries of the matrix
$$ V_n(f)=\left( \begin{matrix} 0 & 0 & ... & ... & f\\ 1 & 0 & 0 & ... & 0\\ \hdotsfor{5}
 \\ 0 & 0 & 0 & ... 1& 0 \end{matrix} \right) $$
as $(V_n(f))_{ij}$ with $i=(a,b)$ and $j=(a',b')$ where $a,a'=1,\ldots, k$
where $k\times k$ is the dimension of $M(f)$, and $b,b'=1,\ldots, n$.
The second indices $(b,b')$ specify which square block of size $k$
we are considering and the first index $(a,a')$ locates the position in that
square block. Then the entries of $V_n(f)$ are
$$ (V_n(f))_{ij}=\left\{ \begin{array}{ll} f_{aa'} & b=1, b'=q \\
1 & a=a', b'=b-1 \\ 0 & \text{otherwise.}
\end{array}\right. $$
Thus, $V_n(f^{\oplus q^{\ell}})$ can be indexed by $(V_n(f^{\oplus q^{\ell}}))_{ij}$
with $i=(a,b,r)$ and $j=(a'b',r')$ with $a,a',b,b'$ as above and $r,r'=1,\ldots q^{\ell}$
with entries as above for $r=r'$ and zero otherwise. Clearly, this also indexes 
the entries of $(V_n(f))^{\oplus q^\ell}$.
\endproof

We use the same notation $F_n$ and $V_n$ for the Frobenius and Verschiebung on $\cW^q_0(A)$. 

\smallskip

\begin{prop}\label{divOmegaq}
Let $k$ be an algebraically closed field.
The divisor map \eqref{divmap} extends to a ring isomorphism 
$\delta_q: \cW^q_0(k) \to \Z[q][k^*]$ given by 
\begin{equation}\label{divcWq}
\delta_q(\Omega_{q^\ell}(E,f))=q^\ell \delta(L(E,f))= q^\ell \sum n(\alpha) [\alpha].
\end{equation}
\end{prop}

\proof We have
$$ \delta_q(\Omega_{q^\ell}(A,f))=\delta(S^{q^\ell}(A,f))= \delta(\prod (1-\alpha t)^{-q^\ell n(\alpha)}) $$
$$ = q^\ell \sum n(\alpha) [\alpha]. $$
It is clear that $\delta_q$ is compatible with addition, and by \eqref{prodOmegaq} it is also
compatible with multiplication
$$ \delta_q(\Omega_{q^{\ell+\ell'}}(E_1\otimes E_2,f_1\otimes f_2))= 
\delta_q(\Omega_{q^\ell}(E_1,f_1) \otimes \Omega_{q^{\ell'}}(E_2,f_2)) $$ $$ =
\delta(S^{q^\ell}(E_1,f_1)\star S^{q^{\ell'}}(E_2,f_2))=\delta(S^{q^{\ell+\ell'}}(E_1\otimes E_2,f_1\otimes f_2))) $$
$$ = q^{\ell+\ell'} \sum n(\alpha)n(\beta) [\alpha][\beta] =
(q^\ell \sum n(\alpha) [\alpha])(q^{\ell'} \sum n(\beta) [\beta])  \in \Z[q][k^*] .$$
The fact that the original $\delta$ is a bijection also implies that $\delta_q$ is a bijection.
\endproof

Using as in \cite{CoCo} an isomorphism $\sigma: \bar \F_p^* \to (\Q/\Z)^{(p)}$ together with
the divisor map $\delta_q$ of \eqref{divcWq} we obtain an isomorphism
\begin{equation}\label{isosigmaq}
\tilde \sigma_q: \cW_0^q(\bar\F_p) \to \Z[q][(\Q/\Z)^{(p)}],  \ \ \  \tilde\sigma_q=\sigma\circ \delta_q.
\end{equation}
Note that the construction of the deformation $\cW^q_0(\bar\F_p)$ and the divisor map
\eqref{isosigmaq} make sense for an arbitrary integer $q$, which is not necessarily a power of $p$. 

\smallskip
\subsection{The $q^\ell$-points deformation of the Bost--Connes algebra}

Let $\tilde \sigma_q: \cW^q_0(\bar F_p) \to \Z[q][(\Q/\Z)^{(p)}] \subset \Z[q][\Q/\Z]$ be as
in \eqref{isosigmaq}. We consider endomorphisms $\sigma_{n,q}$ of $\Z[q][\Q/\Z]$  that
satisfy
\begin{equation}\label{sigmanq}
\sigma_{n,q} \circ \tilde\sigma_q = \tilde\sigma_q \circ F_n ,
\end{equation}
with $F_n$ the Frobenius on $\cW^q_0(\bar \F_p)$.

\begin{lem}\label{sigmaqnendo}
Let $E(r,k):= q^k e(r)$ in $\Z[q][\Q/\Z]$, where $e(r)$ are the generators of $\Z[\Q/\Z]$.
The endomorphisms $\sigma_{n,q}: \Z[q][\Q/\Z] \to  \Z[q][\Q/\Z]$
satisfying \eqref{sigmanq} are given by 
\begin{equation}\label{sigmaqnerk}
\sigma_{n,q}(E(r,k))=q^{k} e(nr) =E(nr,k).
\end{equation}
\end{lem}

\proof Given $\Omega_{q^k}(E,f)\in \cW^q_0(A)$, with $L(E,f)=\prod (1-\alpha t)^{-n(\alpha)}$
we have $\delta_q( \Omega_{q^k}(E,f) )=q^k \sum n(\alpha) [\alpha]$,
while $\delta_q(F_n (\Omega_{q^k}(E,f)) =q^{k} \sum n(\alpha) [\alpha^n]$. 
\endproof

\smallskip

Thus, we see that, with this choice of $q^\ell$-points deformation $\cW^q_0(\bar F_p)$ the
resulting $q$-deformation of the Bost--Connes algebra is essentially trivial, consisting only
of replacing the coefficient ring $\Z$ of $\Z[\Q/\Z]$ by the polynomial ring $\Z[q]$, but the
operations $\sigma_{n,q}$ (hence also the corresponding operations $\tilde\rho_{n,q}$)
remain the same unperturbed operation of the original Bost--Connes algebra and are
the identity on the polynomial ring $\Z[q]$.

\begin{cor}\label{qdef1}
The deformation of the Bost--Connes algebra $\cA_\Z$ induced by the $q^\ell$-points
deformation $\cW^q_0(\bar F_p)$ of the Witt ring $W_0(\bar \F_p)$ is simply given by
the product $\cA_\Z \otimes_\Z \Z[q]$.
\end{cor}

\smallskip
\subsection{The $\A^\ell$-deformation} 

We improve on the construction described above by replacing the 
$q^\ell$-points deformation with another geometrically motivated
deformation, which arises from thinking of the $q^\ell$ points
as the $\F_q$ points of an affine space $\A^\ell$ and constructing
the deformation determined geometrically by taking products with
affine spaces. 

\smallskip

Recall that, if $X$ is a scheme of finite type over $k=\F_q$ with $q=p^r$ a prime power, then
the associated zeta function is given by
$$ Z(X,t) =\exp(\sum_{m\geq 1} N_m(X) \frac{t^m}{m} ) = \prod_x (1-t^{\deg(x)})^{-1}, $$
with $N_m(X)=\# X(F_{q^m})$ and $\deg(x)$ the degree of 
the extension $[k(x):\F_q]$, with $k(x)$ the residue field at the
point $x$. Writing $N_m(X)=\sum_{r|m} r\cdot a_r$, where $$a_r=\#\{ x\,:\,
[k(x):\F_q]=r\},$$ we obtain
\begin{equation}\label{ZXt}
 Z(X,t) = \prod_{r\geq 1} (1-t^r)^{-a_r}. 
\end{equation} 

\smallskip

In the case of an affine space $\A^\ell$ we have
$$ Z(\A^\ell,t)=\exp(\sum_m q^{\ell m} \frac{t^m}{m})=(1-q^n t)^{-1}. $$
In terms of \eqref{ZXt} this is 
\begin{equation}\label{ZetaAq}
 Z(\A^\ell,t)= (1-q^n t)^{-1} = \prod_{r\geq 1} (1-t^n)^{-M(q^n,r)}, 
\end{equation} 
where
\begin{equation}\label{Mqnr}
 M(q^n,r)=\frac{1}{r} \sum_{d|r} \mu(d) \, q^{\frac{nr}{d}},
\end{equation}
with $\mu(x)$ the M\"obius function.

\smallskip

Given a scheme $X$ with zeta function $Z(X,t)$,
taking the product $X\times \A^\ell$ gives
\begin{equation}\label{XAellZeta}
 Z(X\times \A^\ell,t)=Z(X,q^\ell t).
\end{equation} 
More generally, as discussed in \cite{Rama}, one should
regard zeta functions $Z(X,t)$ as elements in the Witt ring $W(\Z)$
satisfying
\begin{equation}\label{ZetaProd}
Z(X\times Y, t)=Z(X.t)\star Z(Y,t),
\end{equation}
where $\star$ is the Witt ring product determined by $(1-at)^{-1}\star (1-bt)^{-1}=(1-ab t)^{-1}$.
For a disjoint union $X\sqcup Y$, the zeta functions multiply, where multiplication of
series corresponds to the addition operation $+_w$ in the Witt ring,
\begin{equation}\label{sqcupZeta}
 Z(X\sqcup Y, t)=Z(X,t) \cdot Z(Y,t) = Z(X,t) +_w Z(Y,t).
\end{equation} 
Indeed, a more general inclusion-exclusion formula holds, see \cite{Rama}.

\smallskip

\begin{rem}\label{gendef} {\rm 
Since $M(q,1)=q$, we have
$$ (1-qt)^{-1} = (1-t)^{-q} \cdot \prod_{r>1} (1-t^r)^{-M(q,r)}. $$
The first term $(1-t)^{-q}$ is the kind of $q$-deformation of $(1-t)^{-1}$
that we considered in the previous sections. In this sense, we can
regard $(1-qt)^{-1}$ as a generalization of the $q$-deformation $(1-t)^{-q}$
discussed before, where the previous deformation appears as the order one term.}
\end{rem}

\smallskip

Modeled on the behavior of these geometric zeta functions, we consider
a different construction of a $q$-deformation of the Witt ring $W_0(A)$.
For $(E,f)\in W_0(A)$ with $L(E,f)=\prod (1-\alpha t)^{-n(\alpha)}$, let
\begin{equation}\label{OmegaAell}
\tilde\Omega_{q^\ell}(L(E,f)) := \prod (1-\alpha\, q^\ell\, t)^{-n(\alpha)}.
\end{equation}

\begin{defn}\label{AllqDef}
The $\A^\ell$-perturbation of the Witt ring $W_0(A)$ is the
graded ring $\tilde\cW^q_0(A)$ defined as a set by
\begin{equation}\label{tildecWq}
\tilde\cW^q_0(A) = \oplus_{\ell\geq 0} \tilde\Omega_{q^\ell} (L(W_0(A))),
\end{equation}
with $L: W_0(A) \to \Lambda(A)$ is the characteristic polynomial map
to the undeformed $\Lambda(A)$, with the operations induced
from $W_0(A)$.
\end{defn}

\smallskip

\begin{lem}\label{multipltildecWq}
Let $A=k$ be an algebraically closed field.
The multiplication in $\tilde\cW^q_0(A)$ obtained as above satisfies
\begin{equation}\label{grmultitildeW}
\star: \tilde\Omega_{q^\ell} (L(W_0(A))) \times \tilde\Omega_{q^{\ell'}} (L(W_0(A))) \to \tilde\Omega_{q^{\ell+\ell'}} (L(W_0(A))).
\end{equation}
\end{lem}

\proof It suffices to check the property on elements of the form $L(E,f)=\prod (1-\alpha t)^{-n(\alpha)}$.
With the undeformed product in $\Lambda(A)$ we have
$$ \prod (1-\alpha\, q^\ell\, t)^{-n(\alpha)} \star \prod (1-\beta\, q^{\ell'}\, t)^{-n(\beta)} =
\prod (1-\alpha \beta\, q^{\ell+\ell'} t)^{-n(\alpha)n(\beta)}. $$
\endproof

\smallskip

\begin{prop}\label{divtildeWq}
Let $k$ be an algebraically closed field. The divisor map \eqref{divmap} extends to a ring isomorphism
$\tilde\delta: \tilde\cW^q_0(k)\to \Z[q][k^*]$ given by
\begin{equation}\label{divtildeq}
\tilde\delta (\tilde\Omega_{q^\ell}(L(E,f)))=\delta (\prod (1-\alpha\, q^\ell\, t)^{-n(\alpha)})=q^\ell \sum n(\alpha) [\alpha].
\end{equation}
Combining this with an isomorphism $\sigma: \bar\F_p^* \to (\Q/\Z)^{(p)}$ as in \cite{CoCo} gives an
isomorphism 
\begin{equation}\label{tildesigmaWtilde}
\hat\sigma_q=\sigma\circ \tilde\delta : \tilde\cW^q_0(\bar\F_p^*) \to \Z[q][(\Q/\Z)^{(p)}].
\end{equation}
\end{prop}

\proof The argument is exactly as in the previous case and in Proposition 2.3 of \cite{CoCo}.
The compatibility with multiplication follows from the previous Lemma.
\endproof

\smallskip
\subsection{$\A^\ell$-deformed Bost--Connes endomorphisms}\label{AellBCsec}

Let $\hat\sigma_q : \tilde\cW^q_0(\bar\F_p^*) \to \Z[q][(\Q/\Z)^{(p)}]$ be as in \eqref{tildesigmaWtilde},
with $\Z[q][(\Q/\Z)^{(p)}]\subset \Z[q][\Q/\Z]$.
We construct endomorphisms $\hat\sigma_{n,q}$ of the ring $\Z[q][(\Q/\Z)^{(p)}]$ that satisfy
\begin{equation}\label{hatsigmanqFn}
\hat\sigma_{n,q} \circ \hat\sigma_q = \hat\sigma_q \circ F_n,
\end{equation}
with $F_n$ the Frobenius on $\tilde\cW^q_0(\bar\F_p^*)$.

\begin{lem}\label{hatsigmalem}
Let $E(r,k)=q^k\, e(r) \in \Z[q][\Q/\Z]$, with $e(r)$ the generators of $\Z[\Q/\Z]$. The endomorphisms
\begin{equation}\label{hatsigmanq}
\hat\sigma_{n,q} (E(r,k))= q^{nk} e(nr) = E(nr,nk)
\end{equation}
satisfy the compatibility condition of \eqref{hatsigmanqFn}.
\end{lem}

\proof It suffices to check that the Frobenius action on $\tilde\cW^q_0(k)$ is given by
\begin{equation}\label{FntildeWq}
F_n\, \Omega_{q^\ell}(L(E,f))= \Omega_{q^{n\ell}} \, F_n L(E,f),
\end{equation}
since we have
$$ F_n \prod (1-\alpha q^\ell t)^{-n(\alpha)} =\prod (1-\alpha^n q^{n\ell} t)^{-n(\alpha)} =
\Omega_{q^{n\ell}} \, L(E,f^n). $$
The divisor map then gives $\tilde\delta(F_n\, \Omega_{q^\ell}(L(E,f)))=q^{n\ell} \sum n(\alpha) [\alpha^n]$.
For $k=\bar\F_p^*$,
this is the image of $\tilde\delta(\Omega_{q^\ell}(L(E,f)))=q^\ell \sum n(\alpha) [\alpha]$ under an 
endomorphism $\hat\sigma_{n,q}$ that induces \eqref{hatsigmanq} on $\Z[\Q/\Z]$.
\endproof

\smallskip

The compatibility with the Verschiebung $V_n$ is more subtle. 

\smallskip

In geometric terms, that is, for elements of the Witt ring that are zeta function of a
scheme of finite type over $k$, the Verschiebung corresponds to the Weil restriction
of scalars. Namely, if $X$ is a scheme of finite type over an extension $k'=\F_{p^{nr}}$ of $k=\F_{p^r}$,
the Verschiebung acts on the zeta function $V_n\, Z(X,t) = Z(R_n X,t)$, 
where $R_n X ={\rm Res}_{k'/k} X$ is the scheme over $k$ obtained by restriction of scalars
from $k'$ to $k$, see \cite{Rama} for more details. 

\smallskip

In our setting, we have the following compatibility condition between Verschiebung
and the deformations $\tilde\Omega_{q^\ell}$.

\begin{lem}\label{VnWtilde}
Let $A=k$ be an algebraically closed field. The Verschiebung $V_n$ on $W_0(A)$ satisfies
\begin{equation}\label{VnOmegaqtilde}
V_n \circ \tilde\Omega_{q^{n\ell}} = \Omega_{q^\ell} \circ V_n.
\end{equation}
\end{lem}

\proof For $(E,f)\in W_0(A)$ with $L(E,f)=\prod (1-\alpha t)^{-n(\alpha)}$ we have
$$V_n(L(E,f)) =L( V_n (E,f)) = L( (E^{\oplus n}, V_n(f))$$ with
$$ V_n(f)=\left( \begin{matrix} 0 & 0 & ... & ... & f\\ 1 & 0 & 0 & ... & 0\\ \hdotsfor{5}
 \\ 0 & 0 & 0 & ... 1& 0 \end{matrix} \right), $$
 hence $L( V_n (E,f))=\prod (1-\omega_{i,n} t)^{-n(\alpha)}$, where $\omega_{i,n}^n =\alpha$ 
 for $i=1,\ldots,n$. Thus, we have 
$$  \Omega_{q^\ell} \, V_n (L(E,f)) = \prod (1-\omega_{i,n}\, q^\ell\, t)^{-n(\alpha)} =
V_n ( \prod (1-\alpha\, q^{n\ell}\, t)^{-n(\alpha)} ) =V_n \, \Omega_{q^{n\ell}} (L(E,f)), $$
hence $V_n$ satisfies \eqref{VnOmegaqtilde}.
 \endproof
 
 \smallskip
 
This shows that, in order to extend the Verschiebung $V_n$ to the entire $\tilde\cW^q_0(A)$,
compatibly with its action on $W_0(A)$ and with the relation \eqref{VnOmegaqtilde}, one
needs to extend the deformation
$$ \tilde\Omega_{q^\ell} : \prod (1-\alpha t)^{-n(\alpha)} \mapsto \prod (1-\alpha\, q^\ell t)^{-n(\alpha)} $$
to include {\em rational} powers of $q$, so that $V_n \circ \tilde\Omega_{q^\ell}$ can
be defined, compatibly with \eqref{VnOmegaqtilde}, when $n$ does not divide $\ell$. 
In terms of our geometric interpretation, these
$q$-deformations arise from products with affine spaces. Equivalently, in motivic
terms, they are given by products with powers of the Lefschetz motive $\bL$. Thus, we can frame 
an appropriate extension of the action of the Verschiebung in terms of the roots of Tate motives 
discussed in \cite{LoMa}.

\smallskip
\subsection{Roots of Tate motives}\label{TateSec}

In the geometric setting, the zeta functions $Z(X,t)$ of schemes of finite type over $k=\F_{p^r}$
determine a ring homomorphism from the Grothendieck ring $K_0(\cV_{\F_{p^r}})$ of varieties 
over $\F_{p^r}$ to the Witt ring $W(\Z)$, see Theorem 2.1 of \cite{Rama}. The Grothendieck
ring $K_0(\cV_k)$ is generated by isomorphism classes $[X]$ of schemes of finite type over $k$
with relations $[X]=[Y]+[X\smallsetminus Y]$ for a closed subscheme $Y$ of $X$ and $[X]\cdot [Y]=[X\times Y]$.
Classes in $K_0(\cV_k)$ are sometimes referred to as ``virtual motives". 
The Lefschetz motive $\bL$ is the class of the affine line $\bL=[\A^1]$ in $K_0(\cV_k)$. 
The Tate subring $\Z[\bL]$ of $K_0(\cV_k)$ maps to the subring of the Witt ring
generated by the zeta functions $Z(\A^\ell,t)=(1-q^\ell t)^{-1}$. 

\smallskip

The idea of introducing roots of Tate motives was first suggested in \cite{Manin2}. A geometric
construction of a square root $\Q(1/2)$ of the Tate motive $\Q(1)$ in terms of supersingular 
elliptic curves was given in \cite{Rama2}. A categorical construction of a square root of the
Tate motive was given in \S 3.4 of \cite{KoSo} and generalized to arbitrary roots $\Q(r)$,
for $r\in \Q_+$, in \cite{LoMa}. We recall here briefly this formal categorical definition. 
Let $\cT={\rm Num}_\Q(k)$ be the Tannakian category of numerical pure motives over $k$
with motivic Galois group $G=G(\cT)$. Consider the homomorphisms $\sigma_n: \bG_m\to \bG_m$
given by $\sigma_n: \lambda\mapsto \lambda^n$, and consider the fibered product of $G$ and $\bG_m$
$$ G^{(n)}=\{ (g,\lambda)\in G\times \bG_m\,:\, t(g)=\sigma_n(\lambda) \}, $$
where $t: G \to \bG_m$ is the group homomorphism that corresponds to the inclusion of the
Tate motives (with Galois group $\bG_m$) inside ${\rm Num}_\Q(k)$. The group $G^{(n)}$
is in turn the Galois group of a Tannakian category $\cT(\Q(\frac{1}{n}))$ which extends the
category $\cT$ by an $n$-th root $\Q(1/n)$ of the Tate motive $\Q(1)$. A more precise description
of these categories and their properties is given in \S 4 of \cite{LoMa}, where the construction of
the roots $\Q(1/n)$ is explained in terms of the primary decomposition of $n\in \N$. Let $\tilde\cT$ 
be the Tannakian category that corresponds to the projective limit $\tilde G$ of the groups $G^{(n)}$, 
as in \cite{LoMa}, that is, the Tannakian category obtained from ${\rm Num}_\Q(k)$ by adjoining
roots of Tate motives of arbitrary order, $\Q(r)$ with $r\in \Q_+$.

\smallskip

At the level of the Grothendieck ring $K_0(\cV_k)[\bL^{-1}]$ with $\bL^{-1}$ the
class of the Tate motive $\Q(1)$, the roots $\Q(1/n)$ correspond to additing new
generatots $\bL^{1/n}$ and their inverses. The Tate part of the resulting Grothendieck
ring $K_0(\tilde\cT)$ would then consist of the ring $\Z[\bL^r\,:\, r\in \Q_+]$, since 
$K_0(\cT(\Q(\frac{1}{n})))=K_0(\cT)[s]/(s^n-\bL)$. 

\begin{lem}\label{Zetaqr}
The ring homomorphism $Z: K_0(\cV_k) \to W(\Z)$ given by the
zeta functions $Z(X,t)$ extends uniquely to a ring homomorphism $Z: K_0(\tilde\cT)\to W(K)$,
obtained by setting $Z(\bL^r,t):=(1-q^r\, t)^{-1}$ for $r\in \Q_+$. Here $K$ is a field that 
contains $\Z$ and all the non-negative real roots $q^r$, for $r\in \Q_+$.
\end{lem}

\proof By the description above of the Grothendieck ring $K_0(\cT(\Q(\frac{1}{n})))=K_0(\cT)[s]/(s^n-\bL)$,
in order to extend the zeta function homomorphism we need to assign to the additional
generator $s$ an element $Z(s)$ in the Witt ring with the property that the $n$-fold product $Z(s)^{\star n}=[1]_w$,
where the unit $[1]_w$ in the Witt ring is $[1]_w :=(1-t)^{-1}$. Assigning $Z(s,t)=(1-q^{1/n}\, t)^{-1} \in W(K)$ satisfies
this requirement. 
\endproof

\smallskip
\subsection{Deformed Bost--Connes algebra}

We then extend the $\A^\ell$-deformation described above, by including roots $\bL^r$, $r\in \Q_+$ 
of the Lefschetz motive $\bL^r$, as explained above. This will have the effect of rendering the
endomorphisms $\sigma_{n,q}$ invertible in their action on the $q$ variable, while the action
on the generators $e(r)$ of $\Q/\Z$ remains unchanged.
The following elementary fact explains the main idea.

\begin{lem}\label{ringZqr}
Consider the endomorphisms $\sigma_n :\Z[q]\to \Z[q]$ determined by $\sigma_n: q\mapsto q^n$.
The direct limit is given by
\begin{equation}\label{cRring}
 \cR= \varinjlim_n \left( \sigma_n : \Z[q]\to \Z[q] \right)  = \Z[q^r\,:\, r\in \Q_+],
\end{equation}
the polynomial ring in the fractional powers $q^r$. The induced endomorphism $\sigma_n$
on the limit $\Z[q^r\,:\, r\in \Q_+]$ has inverse $\rho_n: q^r \mapsto q^{r/n}$.
\end{lem}

\smallskip

We consider the following modification of the deformed Witt ring $\tilde\cW^q_0(A)$.

\begin{defn}\label{hatWq0Adef}
For $r\in \Q_+$. For $A=k$ an algebraically closed field and $(E,f)\in W_0(A)$ with
$L(E,f)=\prod (1-\alpha t)^{-n(\alpha)}$, let
\begin{equation}\label{Omegaqr}
\tilde\Omega_{q^r}(L(A,f))=\prod (1-\alpha\, q^r \, t)^{-n(\alpha)},
\end{equation}
where we treat $q$ and $q^r$ as formal variables.
We set 
\begin{equation}\label{hatWq0r}
\hat\cW^q_0(A)= \oplus_{r\in \Q_+} \Omega_{q^r}(L(W_0(A)),
\end{equation}
with the operations induced from $W_0(A)$.
\end{defn}

As in Lemma~\ref{multipltildecWq}, we have
\begin{equation}\label{grmulthatW}
\star: \tilde\Omega_{q^r} (L(W_0(A))) \times \tilde\Omega_{q^{r'}} (L(W_0(A))) \to \tilde\Omega_{q^{\ell+r'}} (L(W_0(A))).
\end{equation}

\begin{lem}\label{FnVnhatWq0lem}
The Frobenius and Verschiebung on $W_0(A)$ extend to $\hat\cW^q_0(A)$ by
\begin{equation}\label{FnVnhatWq0}
 F_n \circ  \tilde\Omega_{q^r} = \tilde\Omega_{q^{nr}} \circ F_n \ \ \   \text{ and } \ \ \ 
V_n \circ  \tilde\Omega_{q^r} = \tilde\Omega_{q^{r/n}} \circ V_n .
\end{equation}
\end{lem}

\proof The case of Frobenius is as in \eqref{FntildeWq} and the argument for
the Verschiebung is analogous to Lemma~\ref{VnWtilde}.
\endproof

\begin{lem}\label{divZqr}
The divisor map \eqref{divmap} induces a ring isomorphism $\hat\delta: \hat\cW^q_0(k)
\to \cR[k^*]$, with the ring $\cR$ as in \eqref{cRring}.
Combining this with an isomorphism $\sigma: \bar\F_p^* \to (\Q/\Z)^{(p)}$ as in \cite{CoCo} gives an
isomorphism 
\begin{equation}\label{hatsigmaWhat}
\hat\sigma_q=\sigma\circ \hat\delta : \hat\cW^q_0(\bar\F_p^*) \to \cR[(\Q/\Z)^{(p)}].
\end{equation}
\end{lem}

\proof The argument is exactly as in Proposition~\ref{divtildeWq}.
\endproof

\begin{prop}\label{deformedBC}
Let $k=\bar\F_p$.
The $q$-deformed integral Bost--Connes algebra $\cA_{\Z,q}$ determined by the deformation \eqref{hatWq0r}
of the Witt ring $W_0(k)$ is generated by the group ring $\cR[\Q/\Z]$, with $\cR$ as in  \eqref{cRring}
together with elements $\tilde\mu_n$ and $\mu_n^*$ satisfying the relations
\begin{equation}\label{qBCrel1}
\begin{array}{c}
 \tilde\mu_n E(r,r') \mu_n^* = \hat\rho_{n,q}(E(r,r')), \\[3mm]
 \mu_n^* E(r,r') = \hat\sigma_{n,q}(E(r,r')) \mu_n^* \\[3mm]
 E(r,r') \tilde \mu_n =  \tilde \mu_n \hat\sigma_{n,q}(E(r,r')), 
\end{array}
\end{equation} 
where $E(r,r')=q^r\, e(r')$ are the generators of $\cR[\Q/\Z]$ and
$\hat\sigma_{n,q}:\cR[\Q/\Z] \to \cR[\Q/\Z]$ are the endomorphisms
$$ \hat\sigma_{n,q} (E(r,r'))= E(nr,nr'). $$
The $\hat\rho_{n,q}$ are defined by 
\begin{equation}\label{qBCrel2}
 \hat\rho_{n,q}(E(r,r'))= \sum_{ns=r'} E(\frac{r}{n}, s). 
\end{equation} 
These satisfy the relations
\begin{equation}\label{qBCrel3}
\begin{array}{cc} 
\tilde\mu_{nm} = \tilde \mu_n \tilde \mu_m, \ \ m, n\in \mathbb{N}, &
\mu_{nm}^* = \mu_n^* \mu_m^*, \ \ m, n\in \mathbb{N}, \\[3mm]
 \mu_n^*\tilde \mu_n = n, \ \   n \in \N, & 
 \tilde \mu_n \mu_n^* = \mu_n^*\tilde \mu_n , \ \  (n,m) = 1, 
\end{array}
\end{equation} 
 as in the original integral Bost--Connes algebra.
\end{prop}

\proof The endomorphisms $\hat\sigma_{n,q}$ are constructed so as to satisfy
the compatibility with the Frobenius map $F_n$ acting on $\hat\cW^q_0(k)$,
$$ \hat\sigma_{n,q} \circ \hat\sigma_q = \hat\sigma_q \circ F_n, $$
with $\hat\sigma_q$ as in \eqref{hatsigmaWhat}. This determines $\hat\sigma_{n,q}$
to be of the form
$$ \hat\sigma_{n,q} (E(r,r'))= q^{nr} \, e(nr') = E(nr,nr'). $$
In turn the $\hat\rho_{n,q}$ are constructed so as to be compatible with the Verschiebung $V_n$ on
$\hat\cW^q_0(k)$,
$$ \hat\rho_{n,q} \circ \hat\sigma_q = \hat\sigma_q \circ V_n. $$
This determines the $\hat\rho_{n,q}$ to be given by 
$$ \hat\rho_{n,q} (E(r,r'))= q^{r/n} \sum_{ns=r'} e(r')= \sum_{ns=r'} E(\frac{r}{n}, s). $$
\endproof

\begin{lem}\label{AQq}
Let $\cR_\Q=\cR\otimes_\Z \Q=\Q[q^r\,:\, r\in \Q_+]$. The rational $q$-deformed Bost--Connes
algebra $\cA_{\Q,q}$ is the semigroup crossed product $\cR_\Q[\Q/\Z]\rtimes_\rho \N$ 
with generators $E(r,r')=q^r e(r')$ and $\mu_n$, $\mu_n^*$ and with the semigroup crossed product action
given by 
\begin{equation}\label{qBCrel4}
 \mu_n E(r,r') \mu_n^* = \frac{1}{n} \sum_{ns =r'} E(\frac{r}{n},s) =\rho_{n,q} (E(r,r')), 
\end{equation} 
and the relations $E(r,r') E(s,s')=E(r+s,r'+s')$ and 
\begin{equation}\label{munrels}
\begin{array}{ll}
 \mu_n^*\mu_n=1, \,\forall n, & \mu_{nm}=\mu_n\mu_m,\, \forall n,m, \\[2mm]
 \mu_n^*\mu_m^*=\mu_m^*\mu_n^*, \, \forall n,m, &
\mu_n \mu_m^* =\mu_m^* \mu_n, \, \text{ for } (n,m)=1. 
\end{array}
\end{equation}
\end{lem}

\proof This follows directly from the integral algebra $\cA_{\Z,q}$ by taking $\rho_n=\frac{1}{n} \tilde\rho_n$.
The relation between the rational and the integral algebra is as in the original case,
see \cite{CCM3}. 
\endproof

\smallskip
\subsection{The role of the $q$-integers} 

Intuitively, one would expect the $q$-deformation of the Bost--Connes algebra to replace
the integers $n\in \N$ with the $q$-integers $[n]_q=1+q+\cdots+q^{n-1}$. However, it is
clear that this cannot be done just directly, since the $q$-integers $[n]_q$ do not behave
well with respect to the semigroup property that is crucial to the structure of the Bost--Connes algebra.

\smallskip

However, one can see a geometric form of the $q$-integers in the deformed Witt rings $\tilde\cW_0^q(A)$
and $\hat\cW^q_0(A)$, by comparing it again with the case of the zeta functions of schemes over finite fields.
We have seen that mapping $L(E,f)=\prod (1-\alpha t)^{-n(\alpha)}$ to 
$\tilde\Omega_{q^\ell}(L(E,f))=\prod (1-\alpha q^\ell t)^{-n(\alpha)}$ corresponds, in the
case of zeta functions, to mapping $Z(X,t)$ to $Z(X\times \A^\ell,t)$. Similarly, we can consider,
for a scheme $X$ of finite type over $k$, the transformation that maps $Z(X,t)$ to 
$$ Z(X\times \P^n,t) = Z(X,t) \star Z(\P^n,t). $$
The zeta function of $\P^n$ is given by 
$$  Z(\P^n,t) = \prod_{i=0}^n (1-q^i\, t)^{-1} = [1]_w +_w [q]_w +_w [q]_w^{\star 2} +_w \cdots +_w [q]_w^{\star n}, $$
where we denote by $[a]_w:=(1-a t)^{-1}$ in the Witt ring. Thus $Z(\P^n,t)$ corresponds to the $q$-integer $[n+1]_q$
where sum and multiplication are replaced by the corresponding sum an multiplication in the Witt ring. 
This simply reflects the decomposition $\P^n =\A^0\sqcup \A^1   \sqcup \cdots \sqcup \A^n$ and the
corresponding counting of points $\#\P^n(\F_q)=[n+1]_q$. 
The lack of semigroup structure of the $q$-analogs $[n]_q$ corresponds geometrically to
the fact that a product of projective spaces $\P^n \times \P^m$ embeds in a projective 
space $\P^{(n+1)(m+1)-1}$ via the Segre embedding, but is not itself a projective space. 
If the zeta function is given by $Z(X,t)=\prod_{r\geq 1} (1-t^r)^{-a_r}$, then the map above is given by
$$ Z(X,t) \star [n]_q = \prod_{\substack{i=0,\ldots, n-1 \\  r\geq 1}} (1-q^i t^r)^{-a_r}. $$
In a similar way, we have an operation in $\tilde\cW_0^q(A)$
and $\hat\cW^q_0(A)$: for $L(E,f)=\prod (1-\alpha t)^{-n(\alpha)}$ we have
$$ \Omega_{q^r}( L(E,f)) \mapsto \Omega_{q^r}( L(E,f) ) \star [n]_q = \prod_{\substack{i=0,\ldots, n-1 \\  \alpha} }
(1-\alpha q^{r+i} t)^{-n(\alpha)}.  $$

\smallskip
\subsection{Orbit categories and the Habiro ring}\label{HabiroSec}

Let $\cC$ be an additive category and $F$ a self-equivalence.
The {\em orbit category} $\cC/F$ has objects ${\rm Obj}(\cC/F)={\rm Obj}(\cC)$
and morphisms 
\begin{equation}\label{morOrbCat}
{\rm Hom}_{\cC/F}(X,Y):= \oplus_{k\in \Z} {\rm Hom}_{\cC}(X, F^k(Y)).
\end{equation}
If $\cC$ is symmetric monoidal and $F=-\otimes \cO$ is given by tensoring with
an $\otimes$-invertible object $\cO$ in $\cC$, the orbit category $\cC/_{-\otimes\cO}$
is also symmetric monoidal, see \S 7 of \cite{Tab}.

Let $\cT={\rm Num}_\Q(k)$ be the category of numerical motives, as above and 
let $K_0(\cT)$ be the corresponding Grothendieck ring. Then the Grothendieck 
ring of the orbit
category $\cT_n:=\cT/_{-\otimes\Q(n)}$ can be identified with
$K_0(\cT_n)=K_0(\cT)/(\bL^n-1)$, Proposition~3.6 of \cite{LoMa}. While
the Grothendieck ring of varieties $K_0(\cV_k)$ is not the same as the Grothendieck 
ring of numerical motives, this observation gives an interpretation for the
meaning of the quotient rings $K_0(\cV_k)/(\bL^n-1)$. When introducing formal
roots of Tate motives, one can similarly consider orbit categories 
$\cT(\Q(1/n))/_{-\otimes \Q(m/n)}$ and Grothendieck rings $K_0(\cT(\Q(1/n)))/(\bL^{m/n}-1)$.
When one restricts to considering only Tate motives with Grothendieck ring 
$\Z[\bL]$, introducing roots of Tate motives and taking orbit categories leads, respectively,
to the rings $\Z[\bL^r\,:\, r\in \Q_+]$ and the ring
$$
 \widehat{\Z[\bL]}_\infty :=\varprojlim_N \Z[\bL^r\,:\, r\in \Q_+]/\cJ_N, 
$$ 
where $\cJ_N$ is the ideal generated by the elements $(\bL^r)_N:=(\bL^{rN}-1)\cdots (\bL^r-1)$,
see \cite{LoMa}. Equivalently written at the level of counting functions of $\F_q$-points
\begin{equation}\label{hatZinfty}
\widehat{\Z[q]}_\infty :=\varprojlim_N \Z[q^r\,:\, r\in \Q_+]/\cJ_N.
\end{equation}
As shown in \cite{LoMa} and \cite{Mar-endo}, the ring \eqref{hatZinfty} is the same as the
direct limit  
$$ \widehat{\Z[q]}_\infty = \varinjlim_n ( \sigma_n : \widehat{\Z[q]}\to \widehat{\Z[q]} ) $$
of the morphisms $\sigma_n : \widehat{\Z[q]}\to \widehat{\Z[q]}$ determined
by $\sigma_n : q\mapsto q^n$, where $\widehat{\Z[q]}$ is the Habiro ring
\begin{equation}\label{Habiro}
 \widehat{\Z[q]} = \varprojlim_n \Z[q]/ ( (q)_n ),
\end{equation}
where $(q)_n=(1-q)(1-q^2)\cdots (1-q^n)$.

\smallskip

In terms of the construction of the $q$-deformed Bost--Connes algebra described
in Proposition~\ref{deformedBC} and Lemma~\ref{AQq}, this would correspond to
a version of the $q$-deformed algebra where the coefficient ring $\cR$ of \eqref{cRring}
is replaced by $\widehat{\cR}:=\widehat{\Z[q]}_\infty$.

\begin{prop}\label{HabiroBCq}
The Habiro $q$-deformed integral Bost--Connes algebra is generated by the
group ring $\hat\cR[\Q/\Z]$ together with elements $\tilde\mu_n$ and $\mu_n^*$
satisfying the relations \eqref{qBCrel1}, \eqref{qBCrel2}, \eqref{qBCrel3}. The
associated rational algebra is given by the semigroup crossed product 
$\hat\cA_{\Q,q}=\hat\cR[\Q/\Z]\rtimes_\rho \N$ with the semigroup action
determined by \eqref{qBCrel4}. 
\end{prop}

\proof The compatibility of the operations $\hat\sigma_{n,q}$ and $\tilde\rho_{n,q}$
with passing to the projective limit $\hat\cR =\varprojlim_N \cR/\cJ_N$ can be
shown as in Proposition 2.1 and Lemma 2.3 of \cite{Mar-endo}.
The crossed product action in this case remains a semigroup crossed product by $\N$,
unlike in the case of the algebra considered in \cite{Mar-endo} where it becomes a
group crossed product by $\Q^*_+$. Indeed, while the action of $\hat\sigma_{n,q}$
on the coefficient ring $\hat\cR$ is invertible (like it is on $\cR$ itself), the action on
$\hat\cR[\Q/\Z]$ still acts like the original Bost--Connes endomorphisms on the
generators $e(r)$ of $\Z[\Q/\Z]$, with the partial inverses are given by $\rho_{n,q}$.
\endproof

\smallskip
\subsection{Relation to $\F_1$-geometry}\label{F1sec}

The Habiro ring $\widehat{\Z[q]}$ was proposed in \cite{Man-hab} as a model of
analytic functions over $\F_1$. The integral Bost--Connes algebra, on the other
hand was related to $\F_1$-geometry (in the sense of Kapranov--Smirnov \cite{KaSmi})
in \cite{CCM3}. In \cite{Mar-endo} it was shown that a version of the Bost--Connes algebra
can be constructed based on the Habiro ring, and in \cite{LoMa} this was related to the
orbit categories $\cT(\Q(1/n))/_{-\otimes \Q(m/n)}$.

\smallskip

The $q$-deformed Bost--Connes algebra of Proposition~\ref{deformedBC} and
its variant considered in \S \ref{HabiroSec} above combine in a natural
way the integral Bost--Connes algebra of \cite{CCM3} with the version based
on the Habiro ring of \cite{Mar-endo}. This enriches the algebraic form of 
$\F_1$ and its extensions $\F_{1^m}$ based on roots of unity as in \cite{KaSmi},
encoded in the integral Bost--Connes algebra, by combining it with the analytic version developed
in \cite{Man-hab}. 

\smallskip
\subsection{Quantum Statistical Mechanics}\label{QSMsec}

Consider first the case where the Bost--Connes algebra itself remains
undeformed, as in \S \ref{qBCZSec}. Even on the undeformed Bost--Connes
algebra one can consider interesting $q$-deformed time evolutions, that
relate the resulting quantum statistical mechanical system to some
known $q$-deformations of the Riemann zeta function and polylog
functions. The simplest such construction is based on the $q$-analog
of the Riemann zeta function constructed in \cite{Raw1} 
\cite{Raw2}, in the context of a probabilistic approach based on
Bernoulli trials with variable probability. In this context, one assumes
that the variable $q$ is real with $0\leq q < 1$, as it represents a
probability. With this assumption, the $q$-analog of the Riemann
zeta function is constructed as follows.
For $n\in\N$ with primary decomposition $n=\prod p_i^{a_i}$, let
\begin{equation}\label{vn}
v(n)= \sum_i a_i (p_i-1).
\end{equation}
Notice that the function $v(n)$ defined in this way satisfies
$$ v(nm)=v(n) + v(m). $$
Also, for $n\in \N$ with primary decomposition $n=\prod p_i^{a_i}$
we define
\begin{equation}\label{curlnq}
\{ n \}_q =\prod_i [p_i]_q^{a_i},
\end{equation}
where $[p]_q=1+p+\cdots+p^{n-1}$ is the usual $q$-analog. In other
words the $\{ n \}_q$ are the elements of the multiplicative semigroup
$\N_q$ generated by the $q$-analogs of the primes.
Then one sets
\begin{equation}\label{zetaqsmall}
\zeta_q(s) = \sum_{n=1}^\infty \frac{q^{s \nu(n)}}{\{ n\}_q^s}, \ \ \ \ \text{ with } q<1.
\end{equation}
As shown in \cite{Raw2}, this series is convergent for $s>1$ and reduces to the
Riemann zeta function when $q\to 1$.  It is also shown in \cite{Raw2} that
\eqref{zetaqsmall} has an Euler product expansion
\begin{equation}\label{zetaqsmallEuler}
\zeta_q(s) = \prod_p (1- \frac{q^{s(p-1)}}{[p]_q^s})^{-1} .
\end{equation}
In our setting, we regard $q$ as a positive integer, hence $q>1$. The
above expression then needs to be modified accordingly. We have
$$ \frac{q^{p-1}}{[p]_q} = \frac{q^{p-1}}{1+q+q^2+\cdots +q^{p-1}} =\frac{1}{1+q^{-1}+q^{-2} +\cdots+q^{-(p-1)}} 
=\frac{1}{[p]_{q^{-1}}}. $$
Thus, by mapping $q \mapsto q^{-1}$, the expression $q^{p-1}/[p]_q$ is mapped to $1/[p]_q$ and
the version of the zeta function of \eqref{zetaqsmall} for $q>1$ becomes simply
\begin{equation}\label{zetaq}
\zeta_q(s)=\sum_{n\geq 1} \{ n \}_q^{-s} = \prod_p (1- [p]_q^{-s})^{-1}, \ \ \ \  \text{ with } q>1.
\end{equation}

We proceed by defining the
time evolution $\sigma_{t,q}$ on the original Bost--Connes algebra
$\cA_\C=\cA_\Q\otimes_\Q\C$ and on its $C^*$-algebra completion $\cA_{BC}$
is obtained as follows.
 
\begin{lem}\label{sigmaqtlem}
Setting  $\sigma_{t,q}(e(r))=e(r)$ on the generators $e(r)$ of $\C[\Q/\Z]$ and
\begin{equation}\label{sigmatq}
\sigma_{t,q}(\mu_n)= \{n\}_q^{it} \, \mu_n \ \ \  
\end{equation}
defines a time evolution $\sigma:\R \to {\rm Aut}(\cA_\C)$.
In the Bost--Connes representations 
\begin{equation}\label{BCreps}
 \pi_\alpha:\cA_{BC}\to \cB(\ell^2(\N)), \ \ \ \text{ with } \alpha\in \hat\Z^*, 
\end{equation} 
the time evolution $\sigma_{t,q}$ 
is generated by the Hamiltonian
\begin{equation}\label{Hq}
H_q \epsilon_n = \log\{n\}_q \epsilon_n 
\end{equation}
and has partition function
\begin{equation}\label{qZeta}
Z_q(\beta) = \sum_{n= 1}^\infty  \{ n \}_q^{-\beta},
\end{equation}
which is the $q$-analog zeta function of \eqref{zetaq}.
\end{lem}

\proof The assignment \eqref{sigmatq} defines a time evolution since
$\sigma_{t,q}(\mu_n \mu_m)=\{nm\}_q^{it} \mu_{nm} =
\{n\}_q^{it} \{m\}_q^{it} \mu_n \mu_m$ 
and, together with $\sigma_{t,q}(e(r))=e(r)$, is compatible with the relations in the algebra,
and clearly $\sigma_{t+s}(a)=\sigma_t \sigma_s(a)$. In the Bost--Connes representation
associated to the choice of an element $\alpha\in \hat\Z^*$, we have
\begin{equation}\label{BCreps2}
\begin{array}{c}
 \mu_n \epsilon_m =\epsilon_{nm} \\[3mm]
 \pi_\alpha(e(r))\epsilon_m = \zeta_r^m \, \epsilon_m 
\end{array}
\end{equation} 
where $\zeta_r=\alpha^*(e(r))$, identifying $\alpha$ with a choice
of embedding of the roots of unity $\Q/\Z$ in $\C$, with $\zeta_r$ 
the image of $e(r)$ under this embedding. The Hamiltonian $H_q$ satisfies
$$ e^{it H_q} \pi_\alpha (a) e^{-it H_q}  = \pi_\alpha(\sigma_{t,q}(a)) . $$
It suffices to check this condition on $a=\mu_n$ acting on basis elements $\epsilon_m$,
$e^{it H_q} \mu_n e^{-it H_q} \epsilon_m = \{m\}_q^{-it} 
\{nm\}_q^{it} \epsilon_{nm} =\sigma_{t,q}(\mu_n) \epsilon_m$.
The partition function is then given by
$$ Z_q(\beta)=\Tr(e^{-\beta H_q}) =\sum_n \langle \epsilon_n, e^{-\beta H_q} \epsilon_n \rangle=
\sum_{n= 1}^\infty  \{ n \}_q^{-\beta}. $$
\endproof

\smallskip

A more interesting case is the construction of quantum statistical mechanical 
systems associated to the $q$-deformed Bost--Connes algebra 
of Proposition~\ref{deformedBC} and Lemma~\ref{AQq}.

\begin{lem}\label{RepsBCq}
The Bost--Connes representations \eqref{BCreps}, \eqref{BCreps2} extend
to representations of the deformed algebra $\cA_{\C,q}=\cA_{\Q,q}\otimes_\Q\C$
of Lemma~\ref{AQq} on the Hilbert space $\cH=\ell^2(\N)\otimes \ell^2(\Lambda)$
with $\Lambda=q^{\Q_+}$, given by
\begin{equation}\label{qBCreps}
\begin{array}{c}
\mu_n \, \epsilon_{m,\lambda} = \epsilon_{nm, \lambda^{1/n}} \\[3mm]
\pi_\alpha (E(r,r')) \, \epsilon_{m,\lambda}  = \zeta_{r'}^m \,\, \epsilon_{m, \lambda \cdot q^r} \\[3mm]
\mu_n^* \, \epsilon_{m,\lambda} = \left\{ \begin{array}{ll} \epsilon_{m/n, \lambda^n} & n|m \\[2mm]
0 & n\not| m.
\end{array}\right.
\end{array}
\end{equation}
with $E(r,r')=q^r\, e(r')$ and $\zeta_r=\alpha(e(r))$ for a given $\alpha\in\hat\Z^*$, and 
where $\{ \epsilon_{m,\lambda} \}$ with $m\in\N$ and $\lambda\in \Lambda$ is the standard orthonormal
basis of $\cH$.
\end{lem}

\proof These are the same kinds of representations considered in \cite{MaTa}, in the 
special case where elements $\lambda\in \Lambda$ always have an $n$-th root in
$\Lambda$, so that $\mu_n^* \mu_n$ is the identity and not a projector (see Remark~4.10 
and Proposition~4.15 of \cite{MaTa}). First observe that, if we let the generators $q^r=E(r,0)$
act on $\epsilon_{m,\lambda}$ by $E(r,0): \epsilon_{m,\lambda}\mapsto \epsilon_{m,\lambda\cdot q^r}$,
then we need the isometries $\mu_n$ to act as $\mu_n: \epsilon_{m,\lambda}\mapsto \epsilon_{mn,\lambda^{1/n}}$
(which in turn determines the action of $\mu_n^*$ as in \eqref{qBCreps})
in order to satisfy the relations 
\begin{equation}\label{relEr0}
 \mu_n E(r,0) \mu_n^* =\frac{1}{n} \sum_{ns =0 } E(\frac{r}{n},s) 
\end{equation} 
in $\cA_{\C,q}$. Indeed we have, if $n|m$ 
$$ \mu_n E(r,0) \mu_n^* \epsilon_{m,\lambda}=\mu_n E(r,0) \epsilon_{m/n,\lambda^n}=
\mu_n \epsilon_{m/n,\lambda^n\cdot q^r} = \epsilon_{m,\lambda \cdot q^{r/n}} 
= E(\frac{r}{n},0)\, \epsilon_{m,\lambda}, $$
or zero if $n$ does not divide $m$, which agrees with \eqref{relEr0}. One similary checks
compatibility with the other relations: 
$$ \mu_n E(r,r') \mu_n^* =n^{-1} \sum_{ns =r'} E(\frac{r}{n},s) =\rho_n(E(r,r')) $$
and $\mu_n^* E(r,r') \mu_n =\sigma_n(E(r,r'))=E(nr,nr')$, as well as $E(r,r')E(s,s')=E(r+r',s+s')$
and the relations \eqref{munrels}.
\endproof

In order to construct time evolutions that have a convergent partition function $\Tr(e^{-\beta H})$
for sufficiently large $\beta$, it is convenient to enlarge the algebra by additional ``weight
operators", as in \cite{MaTa}.  We define the resulting algebra as follows.

\begin{defn}\label{extAQq}
The extended $q$-Bost--Connes algebra $\cA^w_{\Q,q}$ is generated by $\cR[\Q/\Z]$
and generators $\mu_n$ and $\mu_n^*$ satisfying the relations \eqref{qBCrel4}, \eqref{munrels},
and additional generators given by the weight operators $\omega_z(\lambda)$ for $\lambda\in\Lambda=q^{\Q_+}$
and $z\in U(1)$ satisfying the relations
\begin{equation}\label{Wrels}
\begin{array}{c}
\omega_z(\lambda_1 \lambda_2)=\omega_z(\lambda_1) \omega_z(\lambda_2), \ \ \ \ \  \omega_z(\lambda^{-1})
=\omega_z(\lambda)^{-1} \\[2mm]
\omega_z(\lambda)\, E(r,r') = E(r,r') \, \omega_z(\lambda), \ \ \ \ \  \omega_z(\lambda)\, \mu_n =
\mu_n\, \omega_z(\lambda)^n, \ \ \ \ \  \mu_n^*\, \omega_z(\lambda) = \omega_z(\lambda)^n \, \mu_n^*.
\end{array}
\end{equation}
\end{defn}

\smallskip

\begin{lem}\label{repBCqw}
Suppose given a group homomorphism $h:\Lambda \to \R^*_+$ and an element $\alpha\in \hat\Z^*$.
The representation  of Lemma~\ref{RepsBCq} mapping $\pi_\alpha: \cA_{\C,q}\to \cB(\ell^2(\N\times \Lambda))$, 
for $\alpha\in \hat\Z^*$, extends to a representation of $\cA^w_{\C,q}=\cA^w_{\Q,q}\otimes_\Q\C$ by setting
\begin{equation}\label{omegaact}
\omega_z(\lambda) \epsilon_{m,\eta} = h(\lambda)^{mz}\, \epsilon_{m,\eta},
\end{equation}
independently of $\alpha\in \hat\Z^*$, 
and letting the other generators $E(r,r')$, $\mu_n$ and $\mu_n^*$ act as in \eqref{qBCreps}.
\end{lem}

\proof The argument is exactly as in Proposition 4.15 and Proposition~4.27 
of \cite{MaTa}, with the only difference that in our
case we have $\mu_n^* \mu_n=1$ since elements $\lambda \in \Lambda$ always have an $n$-th
root $\lambda^{1/n}$ in $\Lambda$.
\endproof

We can then define a time evolution as in \cite{MaTa} in the following way.

\begin{lem}\label{timeBCqw}
Setting 
\begin{equation}\label{sigmatBCqw}
\sigma_t(E(r,r'))=\omega_{-it}(q^r) E(r,r'), \ \ \ \  \sigma_t(\mu_n)=n^{it}\mu_n, \ \ \ \   
\sigma_t(\omega_z(\lambda))=\omega_z(\lambda) 
\end{equation}
defines a time evolution on $\cA^w_{\C,q}$.
In the representations of Lemma~\ref{repBCqw} this time evolution is generated by the
Hamiltonian
\begin{equation}\label{HamBCwq}
 H \, \epsilon_{m,\lambda} = (\log(m) - m \log(h(\lambda))) \, \epsilon_{m,n}.
\end{equation} 
\end{lem}

\proof We check that \eqref{sigmatBCqw} determines a time evolution
$\sigma: \R \to {\rm Aut}(\cA^w_{\C,q})$ as in Lemma~4.23 of \cite{MaTa}.
We have
$$ e^{it H} \pi_\alpha (E(r,r')) e^{-itH} \epsilon_{m,\lambda} = e^{it H} \pi_\alpha (E(r,r'))
m^{-it} h(\lambda)^{m it} \epsilon_{m,\lambda} $$ $$ 
= e^{it H} m^{-it} h(\lambda)^{m it} \zeta_{r'}^m \epsilon_{m,\lambda q^r}
= m^{it} h(\lambda q^r)^{-mit} m^{-it} h(\lambda)^{m it} \zeta_{r'}^m \epsilon_{m,\lambda q^r} = $$
$$ h(q^r)^{-m it} \zeta_{r'}^m \epsilon_{m,\lambda q^r} =\omega_{-it}(h(q^r)) E(r,r') \epsilon_{m,\lambda} =\sigma_t(E(r,r'))
\epsilon_{m,\lambda};
$$
$$ e^{it H} \mu_n e^{-itH} \epsilon_{m,\lambda} = e^{it H} \mu_n  m^{-it} h(\lambda)^{m it} \epsilon_{m,\lambda} 
= e^{it H} m^{-it} h(\lambda)^{m it} \epsilon_{mn,\lambda^{1/n}} $$
$$  =(mn)^{it} h(\lambda^{1/n})^{-mn it} m^{-it} h(\lambda)^{m it} \epsilon_{mn,\lambda^{1/n}} = n^{it} \mu_n \epsilon_{m,\lambda};
$$
$$ e^{it H} \mu_n^* e^{-itH} \epsilon_{m,\lambda} =e^{it H} \mu_n^*  m^{-it} h(\lambda)^{m it} \epsilon_{m,\lambda} 
= e^{it H} m^{-it} h(\lambda)^{m it} \epsilon_{m/n,\lambda^n} $$
$$ = (m/n)^{it} h(\lambda^n)^{it m/n} m^{-it} h(\lambda)^{m it} \epsilon_{m/n,\lambda^n} = n^{-it} \mu_n^* \epsilon_{m,\lambda};
$$
$$ e^{it H} \omega_z(\eta) e^{-itH} \epsilon_{m,\lambda} = m^{it} h(\lambda)^{-mit} 
h(\eta)^{zm} m^{-it} h(\lambda)^{mit} \epsilon_{m,\lambda} =  \omega_z(\eta)  \epsilon_{m,\lambda}.
$$
This shows that the operator $H$ of \eqref{HamBCwq} generates the time evolution \eqref{sigmatBCqw}
in the representation of Lemma~\ref{repBCqw}.
\endproof

The choice of the homomorphism $h: \Lambda \to \R^*_+$ in the
construction of the time evolution \eqref{sigmatBCqw} can then be used to
determine the convergence properties of the partition function of the quantum
statistical mechanical system $(\cA^w_{\C,q},\sigma_t)$, as in \cite{MaTa}.
We adapt the representations described above as in \S 4.4 of \cite{MaTa},
by decomposing $\Lambda$ into a countable union of geometric progressions
and acting on a Hilbert space that is a tensor product of the $\ell^2$-spaces
of these countable subsets. 

\smallskip

\begin{prop}\label{timeAqwZ}
Let $\Lambda=q^{\Q_+}$ and 
let $\lambda_k=q^{1/k}\in \Lambda$. Consider the homomorphism $h: \Lambda \to \R^*_+$ 
determined by $h(1)=1$ and $h(\lambda_k)=[p_k]_q$ where $p_k$ is the $k$-th prime number. 
Consider a representation $\pi_\alpha$ of $\cA^w_{\C,q}$ on the Hilbert space 
$$ \cH=\ell^2(\N) \otimes \bigotimes_k \ell^2(\lambda_k^{\Z_+}) $$
satisfying \eqref{qBCreps} and \eqref{omegaact}. Then the time evolution  \eqref{sigmatBCqw}
is implemented in this representation by the Hamiltonian
$$ H \epsilon_{m,\lambda_k^\ell} =( \log(m) - m \ell \log([p_k]_q)) \epsilon_{m,\lambda_k^\ell} . $$
The partition function is given by
\begin{equation}\label{Zetaqnbeta}
 Z_q(\beta)= \sum_{n\geq 1} \zeta_q(n\beta)\, n^{-\beta} ,
\end{equation} 
where $\zeta_q(s)$ is the $q$-analog zeta function of \eqref{zetaq}. The series \eqref{Zetaqnbeta} 
converges for $\beta>3/2$.
\end{prop}

\proof As in Lemma~\ref{timeBCqw} we have
$$ H \epsilon_{m,\lambda_k^\ell} =( \log(m) - m\ell \log(h(\lambda_k)) ) \epsilon_{m,\lambda_k^\ell}
=( \log(m) - m\ell \log([p_k]_q)) \epsilon_{m,\lambda_k^\ell} . $$
Thus, the partition function gives
$$ Z_q(\beta)=\Tr(e^{-\beta H})=\sum_{n\geq 1}\prod_k \sum_{\ell\geq 0} [p_k]^{-\ell n \beta} n^{-\beta} $$
$$ = \sum_n n^{-\beta} \prod_p (1-[p]^{-n\beta})^{-1} =\sum_n n^{-\beta} \zeta_q(n\beta). $$
Since $q>1$, we have $[p]_q=1+q+\cdots+q^{p-1}\geq p$ and $\{ n \}_q\geq n$, hence 
$\zeta_q(n\beta)\leq \zeta(n\beta)$, with $\zeta(s)$ the Riemann zeta function, so that
the convergence of $Z_q(\beta)$ is controlled by the convergence of $\sum_n n^{-\beta} \zeta(n\beta)$. 
The convergence of this series for $\beta>3/2$ was shown in Theorem~4.30 of \cite{MaTa}.
\endproof

\smallskip
\subsection{Weil numbers as a $q$-deformed Bost--Connes system} 

Finally, we want to mention another related way of constructing a $q$-deformed
Bost--Connes system, which is closely related to the one discussed above. 
In the approach developed in \cite{MaTa}, the data of the Bost--Connes
system, consisting of the group ring $\Q[\Q/\Z]$ with the endomorphisms
$\sigma_n$ and partial inverses $\rho_n$, and the associated 
algebra $\cA_\Q=\Q[\Q/\Z]\rtimes_\rho \N$, can be generalized to
other systems based on data $(\Sigma,\sigma_n)$ endowed with
a Galois action, such that the $\sigma_n$ and partial inverses $\rho_n$
arise from Frobenius and Verschiebung functors on a categorification
${\rm Vect}^{\bar k}_\Sigma(k)$ of pairs $(V,\oplus_{s\in\Sigma}\bar V^s)$
of a finite dimensional $k$-vector space $V$ and a $\Sigma$-grading of
$\bar V=V\otimes \bar k$. The category ${\rm Vect}^{\bar k}_\Sigma(k)$ is
neutral Tannakian with Galois group given by the affine $k$-scheme
${\rm Spec}(\bar k [\Sigma]^G)$, with the group operation induced
from the Hopf structure on $\bar k [\Sigma]^G$. Associated to such
data, endowed with a suitable set of $G$-equivariant embeddings 
of $\Sigma$ to $\bar\Q^*$, it is possible to construct an
algebra $\cA_{(\Sigma,\sigma_n)}$ together with a time evolution $\sigma_t$,
so that the properties of the resulting quantum statistical mechanical system 
generalize the original properties of the Bost--Connes system. In particular,
as shown in Example 4 of \S 5 of \cite{MaTa}, one such generalization
of the Bost--Connes algebra can be constructed for the data $(\Sigma,\sigma_n)$
where $\Sigma = \cW(q)$ is the set of Weil numbers, namely the subgroup of $\bar\Q^*$
given by algebraic numbers $\pi$ such that, 
\begin{itemize}
\item for every embedding $\rho: \Q[\pi]\to \C$ 
one has $|\rho(\pi)|=q^{m/2}$ for some $m\in \Z$, the weight $m=w(\pi)$;
\item there is some integer 
$s$ for which $q^s \pi$ is an algebraic integer. 
\end{itemize}
Under the identification $\cW(q)\simeq \cW_0(q)\times \Z$
by $\pi \mapsto (\frac{\pi}{|\rho(\pi)|}, w(\pi))$ he surjective group homomorphisms
$\sigma_n :\cW(q)\to \cW(q)$ are given by $(\pi,m)\mapsto (\pi^n, nm)$.
Since roots of unity $\Q/\Z$ are contained inside $\cW(q)$ as the
subgroup $\Q/\Z \times \{0\}\subset \cW_0(q)\times \Z$, one can 
regard the resulting Bost--Connes algebra $\cA_{(\cW(q),\sigma_n)}$ 
associated to the datum $(\cW(q),\sigma_n)$, as constructed in \cite{MaTa} 
as another form of $q$-deformation of the original Bost--Connes algebra
$\cA_\Q=\cA_{(\Q/\Z,\sigma_n)}$.  In terms of zeta functions, Weil numbers
correspond to Frobenius eigenvalues of motives over finite fields, \cite{Milne},
hence this can be regarded as another way of using zeta functions as a 
model for $q$-deforming the Bost--Connes algebra. We refer the
readers to \cite{MaTa} for more details.

\bigskip
\subsection*{Acknowledgments}  The
first author was partially supported by NSF grants DMS-1201512
and PHY-1205440. The second author was supported by
a Summer Undergraduate Research Fellowship at Caltech.

\end{document}